%% file: main.tex
\documentclass[conference]{IEEEtran}

\usepackage{amsmath,amsthm,relsize}
\usepackage{amsfonts, amssymb, cuted}
\usepackage{cleveref}
\usepackage{mathtools}
\usepackage[boxed]{algorithm}
\usepackage[noend]{algpseudocode}
\usepackage{cite}
\usepackage{xcolor}
\usepackage{soul}
\usepackage{graphicx}
\usepackage{tikz}
\usepackage{pgfplotstable}
\usepackage{pgfplots}
\usepackage{nicefrac}
\usepackage{enumitem}
\usepackage{support-caption}
\usepackage{subcaption}
\usepackage[font=footnotesize]{caption}
\usepackage{multirow}
\pgfplotsset{compat=1.15}
\usepackage{relsize}
\usetikzlibrary{patterns}
\usepackage{acro}
\usepackage{color,soul}

\usepackage{bbm}

\input{commands}

\input{MyCommands}

\usepackage[letterpaper,left=0.73in,right=0.73in, bottom=1.07in, top=0.73in]{geometry}

\title{Finite-Length Analysis of Wiretap Codes using Universal Hash Functions}

\begin{document}

\author{\IEEEauthorblockN{Kagan Akcay} \\
\IEEEauthorblockA{
     Electrical Engineering and Computer Science Department, Technische Universit\"at Berlin, 10587 Berlin, Germany\\
    \textrm{kagan.akcay@tu-berlin.de}
    }
}

\maketitle

\begin{abstract}
    This paper investigates the relation between the second-order coding rate, where the second-order turns out to be strictly larger than $\sqrt{n}$, and the mutual information as the leaked information for a fixed error probability by using wiretap codes constructed by universal$_2$ hash functions. We first generalize the upper bound on $\epsilon$-smooth max information in \cite{tyagi} and use it in our analysis where we adopt the method in \cite{hayashi-tan}, which uses universal hashing for compressing a source and making it secure from another correlated source, and apply it to the wiretap channel. We prove first- and second-order achievability results by assuming that the conjecture we state holds true.
\end{abstract}

\section{Introduction}

The study of secure communication over a noisy channel with an eavesdropper, which is called a wiretap channel, began with Wyner \cite{wyner} and continued with Csisz\'ar and K\"orner  \cite{generalwiretap}. It was shown that if the transmission rate is below the secrecy capacity, the error probability and the information leakage can be made arbitrarily small as long as the number of channel usage is sufficiently large.  Whereas the studies above analyze the wiretap channel in the asymptotic regime, the analysis in the finite length regime is also studied in several papers such as \cite{yang-schaefer-poor,variation2,hayashi}, where \cite{yang-schaefer-poor,variation2} adopt the leaked information criterion based on the variational distance and \cite{hayashi} on the mutual information. This paper adopts the leaked information criterion based on mutual information.

The main purpose of this paper is to find the relation between the second-order coding rate and the leaked information for a fixed error probability. Since the mutual information can be unbounded, 
we don't fix the leaked information together with the error probability and analyze the asymptotic expansion as done in \cite{yang-schaefer-poor} where the variational distance as the leaked information is bounded. We fix only the error probability and investigate the first- and second-order asymptotics of the leaked information by adopting the method in \cite{hayashi-tan}.~\cite{hayashi-tan} proves first- and second-order results 
for the mutual information between a compressed version of a source $A^n$, where universal hashing is used for compression, and another correlated source $E^n$, where $(A^n,E^n)$ is an independent identical distributed (i.i.d) sequence. We follow the achievability proof method of~\cite{hayashi-tan} 
and apply it to the wiretap channel by proving an upper bound on $\epsilon$-max smooth information. 

In Section~\ref{sec:Sys_mod}, we define the wiretap channel and the secrecy codes. In Section~\ref{sec:types}, we prove new upper bounds on the size of the typical sets and the channel probability of the elements in a typical set to be used in the upper bound on $\epsilon$-max information in Section~\ref{sec:max-info}. In Section~\ref{sec:divergence}, we prove an upper bound on the mutual information between any two random variables taking values on finite sets by $\epsilon$-max information and an error term by proving that the R\'enyi divergence of order $2$ is upper bounded by max information. In Section~\ref{sec:hashing}, we state the definition of a universal$_2$ hash function and define a new kind of a universal$_2$ hash function, which we call a c-universal$_2$ hash function, with the extra condition that the cardinality of the image of its right inverse does not depend on its argument as in~\cite{hayashi}. 
We also state a lemma from \cite{hayashi-tan} which shows a relation 
between the R\'enyi divergence between the compressed version of $A$ and  $E$, and the conditional R\'enyi entropy between $A$ and $E$, where $(A,E)$ is a correlated source, to be used for our finite-length achievability analysis of the wiretap channel in the next section. In Section~\ref{sec:Achievability_bounds}, we prove first- and second-order achievability results for fixed error probability by constructing wiretap codes using c-universal$_2$ hash functions and assuming that certain conjecture holds true. 
The second-order turns out to be strictly larger than $\sqrt{n}$, whereas in \cite{yang-schaefer-poor,variation2,hayashi-tan} it is exactly $\sqrt{n}$.


\section{Preliminaries} 


\begin{defn} 
\cite{hayashi-tan} Let \textbf{A} be a finite set where $P$ and $Q$ are two distributions on $\textbf{A}$. The R\'enyi divergence of order $1+s$ is given by
$$D_{1+s}(P\vert{}\vert{}Q):=\frac{1}{s}\log{}\sum_{a\in{\textbf{A}}}P(a)^{1+s}Q(a)^{-s}$$
for every $s>0$. It is easy to show that $D_1(P\vert{}\vert{}Q):=\lim_{s\rightarrow{0}}D_{1+s}(P\vert{}\vert{}Q)=D(P\vert{}\vert{}Q)$.
\end{defn}

\begin{defn} 
\cite{hayashi-tan} Let $\textbf{A}$ and $\textbf{E}$ be two sets with finite cardinality, let $(A,E)$ be a pair of joitnly distributed random variables on $\textbf{A} \times \textbf{E}$. The conditional R\'enyi entropy of order $1+s$ of $A$ given $E$ is given by
$$H_{1+s}(A\vert{}E):=-\frac{1}{s}\log{}\sum_eP_E(e)\sum_aP_{A\vert{}E}(a\vert{}e)^{1+s}$$
for $s>0$. It is easy to show that $H_1(A\vert{}E):=\lim_{s\rightarrow{0}}H_{1+s}(A\vert{}E)=H(A\vert{}E)$.
\end{defn}



\begin{defn} 
Let \textbf{X} and \textbf{Y} be finite sets, where $P$ is a distribution on \textbf{X} and $W:\textbf{X}\rightarrow{\textbf{Y}}$ is a stochastic matrix. The mutual information is then given by
$$I(P;W):=H(PW)-H(W\vert{}P)$$
where $PW(y)=\sum_{x\in{\textbf{X}}}P(x)W(y\vert{}x)$ for all $y\in{\textbf{Y}}$. If $(X,Y)$ are jointly distributed on $\textbf{X} \times \textbf{Y}$ and $P_{Y\vert{}X}=W$, we write $I(X\wedge{}Y)$ for $I(P_X;W)$.
\end{defn}
%
%
\section{Channel Model}
\label{sec:Sys_mod}
\input{system_model}

\section{Types and Typicality} 
\label{sec:types}
\input{types}

\section{Max-Information}
\label{sec:max-info}
\input{max_info}

\section{Divergence Bounds}
\label{sec:divergence}
\input{renyi_entropy}

\section{Universal Hashing}
\label{sec:hashing}
\input{uni_hashing}

\section{Achievability Bounds on the Wiretap Channel}
\label{sec:Achievability_bounds}
\input{Achievability_bounds}

\section{Conclusion}

We have generalized the upper bound on $\epsilon$-max information and proved an upper bound on the mutual information between any two random variables, taking values on finite sets by $\epsilon$-max information and an error term. We proved first- and second-order achievability results for fixed error probability by constructing wiretap codes using c-universal$_2$ hash functions and assuming a conjecture holds true, where the second-order turns out to be strictly larger than $\sqrt{n}$. 

\bibliographystyle{IEEEtran}
\bibliography{references}

\end{document}

%% file: commands.tex
\newcommand{\vbar}{\raisebox{.17ex}{\rule{.04em}{1.35ex}}}
\newcommand{\vbarind}{\raisebox{.01ex}{\rule{.04em}{1.1ex}}}
\newcommand{\R}{\ifmmode{\rm I}\hspace{-.2em}{\rm R} \else ${\rm I}\hspace{-.2em}{\rm R}$ \fi}
\newcommand{\T}{\ifmmode{\rm I}\hspace{-.2em}{\rm T} \else ${\rm I}\hspace{-.2em}{\rm T}$ \fi}
\newcommand{\N}{\ifmmode{\rm I}\hspace{-.2em}{\rm N} \else \mbox{${\rm I}\hspace{-.2em}{\rm N}$} \fi}
\newcommand{\B}{\ifmmode{\rm I}\hspace{-.2em}{\rm B} \else \mbox{${\rm I}\hspace{-.2em}{\rm B}$} \fi}
\newcommand{\Hil}{\ifmmode{\rm I}\hspace{-.2em}{\rm H} \else \mbox{${\rm I}\hspace{-.2em}{\rm H}$} \fi}
\newcommand{\C}{\ifmmode\hspace{.2em}\vbar\hspace{-.31em}{\rm C} \else \mbox{$\hspace{.2em}\vbar\hspace{-.31em}{\rm C}$} \fi}
\newcommand{\Cind}{\ifmmode\hspace{.2em}\vbarind\hspace{-.25em}{\rm C} \else \mbox{$\hspace{.2em}\vbarind\hspace{-.25em}{\rm C}$} \fi}
\newcommand{\Q}{\ifmmode\hspace{.2em}\vbar\hspace{-.31em}{\rm Q} \else \mbox{$\hspace{.2em}\vbar\hspace{-.31em}{\rm Q}$} \fi}
\newcommand{\Z}{\ifmmode{\rm Z}\hspace{-.28em}{\rm Z} \else ${\rm Z}\hspace{-.28em}{\rm Z}$ \fi}



%% file: MyCommands.tex
\newtheorem{thm}{Theorem}
\newtheorem{lem}{Lemma}

\newtheorem{conj}{Conjecture}
\theoremstyle{definition}
\newtheorem{defn}{Definition}
\newtheorem{rem}{Remark}



%% file: system_model.tex
The wiretap channel was first introduced by Wyner in \cite{wyner}. In the wiretap channel, we have a sender, a legitimate receiver, and an eavesdropper. The sender aims to transmit a message reliably to the legitimate receiver while keeping it secret from the eavesdropper. In the discrete memoryless degraded wiretap channel (DM-DWTC), we have $V^n : \textbf{X}^n\rightarrow{\textbf{Y}^n}$ and $W^n : \textbf{Y}^n\rightarrow{\textbf{E}^n}$ where $V$ and $W$ are discrete memoryless channels (DMCs), \textbf{X}, \textbf{Y} and \textbf{E} are finite sets and $n$ is the number of the channel usage. $V$ corresponds to the channel from the sender to the legitimate receiver, whereas $W$ corresponds to the degraded channel to the eavesdropper. 
A secrecy code $C_n$ for a DM-WTC consists of
\begin{itemize}
\item a message $M$ which distributed on a message set $\textbf{M}=\lbrace1,...,\vert{}\textbf{M}\vert{}\rbrace$,
\item a randomized encoder $e:\textbf{M}\rightarrow{\textbf{X}^n}$ which generates a codeword $x^n$ according to a probability distribution $P_{X^n\vert{}M=m}$,
\item a decoder $d:\textbf{Y}^n\rightarrow{\textbf{M}}$ which maps each channel observation $y^n$ to an estimate $\hat{m}\in{\textbf{M}}$.
\end{itemize}
We require that $M$ is uniformly distributed on $\textbf{M}$ and the encoder and the decoder of the legitimate receiver satisfy the average error probability constraint $Pr(d(Y^n)\neq{}M)\leq{}\varepsilon$. 
%
We assume that the channel statistics and properties of the code $C_n$ are known to all parties. However, the eavesdropper cannot access realizations $x^n$ and $y^n$. We use mutual information $I(M\wedge{}E^n)$ as the secrecy metric. We fix $\varepsilon$ and let the message set $\vert{}\textbf{M}\vert{}=e^{nR+(\sqrt{n})^{1+\gamma}L}$ where $0<\gamma<1$ and $\gamma$ is fixed and analyse the asymptotic behaviour of $I(M\wedge{}E^n)$ with respect to $\varepsilon$, $R$ and $L$. The reason why $\gamma>0$ and not exactly $0$ as in~\cite{yang-schaefer-poor,variation2,hayashi-tan} will be clear in Section~\ref{sec:Achievability_bounds}.

%% file: types.tex
\begin{defn}   \label{type}
\cite[Definition 2.1]{info} The type of a sequence $x^n\in{\textbf{X}^n}$ is the distribution $P_{x^n}$ on $\textbf{X}$ defined by
$$P_{x^n}(a):=\frac{1}{n}N(a\vert{}x^n), \;\; \forall \; a \in \textbf{X}$$
where $N(a|x^n)$ denotes the multiplicity of symbol $a$ in $x^n$. For any distribution $P$ on $\textbf{X}$, the set of sequences of type $P$ in $\textbf{X}^n$ is denoted by $T^n_P$.
\end{defn}
\begin{defn} \label{typical}
\cite[Definition 2.8]{info} For any distribution $P$ on $\textbf{X}$, a sequence $x^n\in{\textbf{X}^n}$ is called $P$-typical with constant $\delta$ if
$$\vert{}N(a\vert{}x^n)/n-P(a)\vert{}\leq{\delta}$$
for every $a\in{\textbf{X}}$ and, in addition, no $a\in{\textbf{X}}$ with $P(a)=0$ occurs in $x^n$. The set of such sequences will be denoted by $T^n_{[P]_\delta}$. If $X$ is a random variable with values in $\textbf{X}$, we refer to $P_X$-typical sequences as $X$-typical, and write $T^n_{[X]_\delta}$ for $T^n_{[P_X]_\delta}$.
\end{defn}
\begin{defn}  \label{conditionaltype}   
\cite[Definition 2.9]{info} For a stochastic matrix $W : \textbf{X}\rightarrow{\textbf{Y}}$, a sequence $y^n\in{\textbf{Y}^n}$ is $W$-typical under the condition $x^n\in{\textbf{X}^n}$ with constant $\delta$ if
$$\left\vert{}\frac{1}{n}N(a,b \vert{}x^n,y^n)-\frac{1}{n}N(a\vert{}x^n)W(b\vert{}a)\right\vert{}\leq{\delta}$$
for every $a\in{\textbf{X}}$, $b\in{\textbf{Y}}$, and, in addition, $N(a,b\vert{}x^n,y^n)=0$ whenever $W(b\vert{}a)=0$. The set of such sequences $y^n$ will be denoted by $T^n_{[W]_\delta}(x^n)$. If $X$ and $Y$ are random variables with values in $\textbf{X}$ resp. $\textbf{Y}$ and $P_{Y\vert{}X}=W$, then we speak of $Y\vert{}X$-typical sequences and write $T^n_{[Y\vert{}X]_\delta}(x^n)$ for $T^n_{[W]_\delta}(x^n)$. 
\end{defn}
\begin{lem}  \label{countinglemma}
\cite[Lemma 2.2]{info} The number of different types of sequences in $\textbf{X}^n$ is less than $(n+1)^{\vert{}\textbf{X}\vert{}}$.
\end{lem}
\begin{lem} \label{deltalemma}
\cite[Lemma 2.10]{info} If $x^n\in{T^n_{[X]_\delta}}$ 
and $y^n\in{T^n_{[Y\vert{}X]_{\delta'}}(x^n)}$ 
then $y^n\in{T^n_{[Y]_{\delta''}}}$ 
for $\delta'':=(\delta+\delta')\vert{}\textbf{X}\vert{}$.
\end{lem} 
\begin{lem} \label{totalprobabilitylemma}
\cite[Lemma 2.12]{info} For every stochastic matrix $W:\textbf{X}\rightarrow{\textbf{Y}}$
$$W^n(T^n_{[W]\delta}(x^n)\vert{}x^n)\geq{1-2\vert{}\textbf{X}\vert{}\vert{}\textbf{Y}\vert{}e^{-2n\delta^2}}$$
for every $\delta>0$.
\end{lem}
\begin{lem}  \label{nsizeofsetlemma}
For every distribution $P$ on $\textbf{X}$
$$\log{}\vert{}T^n_{[P]_\delta}\vert{}\leq{}nH(P)-n\delta\vert{}\textbf{X}\vert{}\log{\min_{a\in{supp(P(a))}}P(a)}$$
for every $\delta>0$. 
\end{lem}
\begin{proof}
Let $x^n\in{T^n_{[P]_\delta}}$. Then
\begin{align}
P^n(x^n)&=\prod_{a\in{supp(P(a))}}P(a)^{N(a\vert{}x^n)} \nonumber \\
                        &\geq{\prod_{a\in{supp(P(a))}}P(a)^{n(\delta+P(a))}} \nonumber \\
                        &=2^{\sum_{a\in{supp(P(a))}}n(\delta+P(a))\log{P(a)}} \nonumber \\
                        &\geq{2^{-nH(P)+n\delta\vert{}\textbf{X}\vert{}\log{\min_{a\in{supp(P(a))}}P(a)}}}. \nonumber
\end{align}
Since $\sum_{x^n\in{T^n_{[P]_\delta}}}P^n(x^n)\leq{}1$, $\vert{}T^n_{[P]_\delta}\vert{}\min_{x^n\in{T^n_{[P]_\delta}}}P^n(x^n)\leq1$ so that
$$\log{}\vert{}T^n_{[P]_\delta}\vert{}\leq{}nH(P)-n\delta\vert{}\textbf{X}\vert{}\log{\min_{a\in{supp(P(a))}}P(a)}.$$
\end{proof} 
\begin{lem} \label{nindividualprobabilitylemma}
Let $W : \textbf{X}\rightarrow{\textbf{Y}}$ be a stochastic matrix,  $x^n\in{\textbf{X}^n}$ and $\delta>0$. Then for all $y^n\in{T^n_{[W]_\delta}(x^n)}$, it holds that
$$\log{W^n(y^n\vert{}x^n)}\leq{-nH(W\vert{}P)+n\delta{}W_c\vert{}\textbf{X}\vert{}\vert{}\textbf{Y}\vert{}},$$
where $P$ is the type of $x^n$ and \\
$W_{c}=-\log{\min_{(a,b)\in{supp(W(b\vert{}a))}}{W(b\vert{}a)}}$.
\end{lem}
\begin{proof}
Let $P$ be the type of $x^n$ and $y^n\in{T^n_{[W]_\delta}(x^n)}$. Then 
\begin{align}
&\log{W^n(y^n\vert{}x^n)}=\log{\prod_{(a,b)\in{supp(W(b\vert{}a))}}W(b\vert{}a)^{N(a,b\vert{}x^n,y^n)}} \nonumber \\
                        &\leq{\log{\prod_{(a,b)\in{supp(W(b\vert{}a))}}W(b\vert{}a)^{(-n\delta+N(a\vert{}x^n)W(b\vert{}a))}}} \nonumber \\
                        &=\log{2^{\sum_{(a,b)\in{supp(W(b\vert{}a))}}(-n\delta+nP(a)W(b\vert{}a))\log{W(b\vert{}a)}}} \nonumber \\
                        &\leq{-nH(W\vert{}P)-n\delta\vert{}\textbf{X}\vert{}\vert{}\textbf{Y}\vert{}\log{\min_{(a,b)\in{supp(W(b\vert{}a))}}{W(b\vert{}a)}}}. \nonumber
\end{align}
\end{proof}

%% file: max_info.tex
\begin{defn} 
\cite[Section 3-B]{tyagi} Let $W : \textbf{X}\rightarrow{\textbf{Y}}$ be a subnormalized channel 
where $\textbf{X}$ and $\textbf{Y}$ are finite sets and $\sum_{y}W(y\vert{}x)\leq{}1$ for every $x\in{\textbf{X}}$. The max-information of $W$ is given by
$$I_{max}(W):=\log{\sum_y\max_{x}w(y\vert{}x)}.$$
\end{defn}
\begin{defn}  \label{epsilonsmoothdefinition}
\cite[Section 3-B]{tyagi} Let $W : \textbf{X}\rightarrow{\textbf{Y}}$ be a channel, where $\textbf{X}$ and $\textbf{Y}$ are finite sets. For a subset $\textbf{T}$ of $\textbf{X}\times{}\textbf{Y}$, denote by $W_{\textbf{T}}$ the subnormalized channel where
$$W_{\textbf{T}}(y\vert{}x)=W(y\vert{}x)\mathbbm{1}{((x,y)\in{\textbf{T}})}.$$
The $\epsilon$-smooth max-information of $W$ is given by
$$I_{max}^{\epsilon}(W):=\inf_{\textbf{T}}I_{max}(W_{\textbf{T}})$$
over all sets $\textbf{T}\subset{\textbf{X}\times{}\textbf{Y}}$ such that
$$W(\lbrace{}y:(x,y)\in{\textbf{T}}\rbrace\vert{}x)\geq{}1-\epsilon$$
for all $x\in{\textbf{X}}$.
\end{defn}
\begin{thm}  \label{maxinfotheorem}
Let $W: \textbf{X}\rightarrow{\textbf{Z}}$ be a DMC with finite input and output alphabets $\textbf{X}$ and $\textbf{Z}$, respectively, and $e_0: \textbf{V}\rightarrow\textbf{X}^n$ be an encoder. Then, denoting $\epsilon_n=2\vert{}\textbf{X}\vert{}\vert{}\textbf{Z}\vert{}e^{-2n\delta^2}$, for the combined channel $W_{e_0}=W^n\circ{}e_0$ it holds that
\begin{align*}
I_{max}^{\epsilon_n}(W_{e_0})\leq{} &n\max_{P_X}I(X\wedge{}Z)+2n\delta{}W_c\vert{}\textbf{X}\vert{}\vert{}\textbf{Z}\vert{}\\
&+\vert{}\textbf{X}\vert{}\log{(n+1)}, 
\end{align*}
for every $\delta>0$ where \\
$W_{c}=-\log{\min_{(a,b)\in{supp(W(b\vert{}a))}}{W(b\vert{}a)}}$.
\end{thm}
\begin{proof}
Let $P$ be a type in $\textbf{X}^n$ as in Definition \ref{type}. Denote by $T^n_{(W)_\delta}$ the set of sequences $(x^n,z^n)$ such that $z^n$ is $W$-typical given $x^n$ with $\delta$ as in Definition \ref{conditionaltype}, by $T^n_{(P,W)_\delta}$ the set of sequences $(x^n,z^n)\in{T^n_{(W)_\delta}}$ such that $x^n$ has type $P$, and by $\Phi[T^n_{(P,W)_\delta}]$ the projection of $T^n_{(P,W)_\delta}$ on $\textbf{Z}^n$. Then by Lemma \ref{nindividualprobabilitylemma} for all $(x^n,z^n)\in{T^n_{(P,W)_\delta}}$
\begin{equation} \label{indiveq}
\log{W^n(z^n \vert{x^n})}\leq{-nH(W\vert{}P)+n\delta{}W_c\vert{}\textbf{X}\vert{}\vert{}\textbf{Z}\vert{}}.
\end{equation}
By Lemma \ref{deltalemma} $\Phi[T^n_{(P,W)_\delta}]\subset{T^n_{[PW]_{\delta\vert{}\textbf{X}\vert{}}}}$ where $PW$ denotes the output distribution for channel $W$ when the input distribution is $P$. Then by Lemma \ref{nsizeofsetlemma} $\log{\vert{\Phi[T^n_{(P,W)_\delta}]}}\vert{}\leq{}nH(PW)-n\delta\vert{}\textbf{X}\vert{}\vert{}\textbf{Z}\vert{}\log{\min_{b\in{supp(PW(b))}}PW(b)}.$ Since $PW(b)=\sum_aP(a)W(b\vert{}a)\geq{}\min_aW(b\vert{}a)$, 
\begin{equation} \label{seteq}
\log{\vert{\Phi[T^n_{(P,W)_\delta}]}}\vert{}\leq{}nH(PW)+n\delta{}W_c\vert{}\textbf{X}\vert{}\vert{}\textbf{Z}\vert{}. 
\end{equation}
Furthermore by Lemma \ref{totalprobabilitylemma} for every $v\in{\textbf{V}}$ 
\begin{equation}
\sum_{z^n:(e_0(v),z^n)\in{T^n_{(W)_{\delta}}}}W^n(z^n\vert{e_0(v)})\geq{1-\epsilon_n}, \nonumber
\end{equation}
where $\epsilon_n=2\vert{}\textbf{X}\vert{}\vert{}\textbf{Z}\vert{}e^{-2n\delta^2}$. First, we assume that each codeword $e_0(v)$ is of type $P$ and $P$ is fixed. For simplicity, we omit the $\delta$ in 
$T^n_{(P,W)_\delta}$. Then for the subnormalized channel $W_{e_0,T^n_{(P,W)}}$
\begin{align} \label{fixedtypeeq}
&I_{max}^{\epsilon_n}(W_{e_0})\leq{I_{max}(W_{e_0,T^n_{(P,W)}})}   \nonumber  \\  
                                      &=\log{\sum_{z^n}\max_vW^n(z^n\vert{e_0(v))}\mathbbm{1}{((e_0(v),z^n)\in{T^n_{(P,W)}})}}        \nonumber \\
                                      &\stackrel{(a)}{\leq{}}-nH(W\vert{P})+\log{\sum_{z^n}\max_v\mathbbm{1}{((e_0(v),z^n)\in{T^n_{(P,W)}})}} \nonumber \\ 
                                      &+n\delta{}W_c\vert{}\textbf{X}\vert{}\vert{}\textbf{Z}\vert{} \nonumber   \\      
                                      &\stackrel{(b)}{\leq{}}-nH(W\vert{P})+\log{\vert{\Phi[T^n_{(P,W)}]}}\vert{}+n\delta{}W_c\vert{}\textbf{X}\vert{}\vert{}\textbf{Z}\vert{} \nonumber \\
                                      &\stackrel{(c)}{\leq{}}-nH(W\vert{}P)+nH(PW)+2n\delta{}W_c\vert{}\textbf{X}\vert{}\vert{}\textbf{Z}\vert{} \nonumber \\
                                      &=nI(P;W)+2n\delta{}W_c\vert{}\textbf{X}\vert{}\vert{}\textbf{Z}\vert{},
\end{align}
where $(a)$ comes from (\ref{indiveq}), $(b)$ uses the fact that $(e_0(v),z^n)\in{T^n_{(P,W)}}$ implies $z^n\in{\Phi[T^n_{(P,W)}]}$ so that $\sum_{z^n}\max_v\mathbbm{1}{((e_0(v),z^n)\in{T^n_{(P,W)}})}\leq{}\sum_{z^n}\mathbbm{1}{(z^n\in{\Phi[T^n_{(P,W)}]})}=\vert{}\Phi[T^n_{(P,W)}]\vert{}$, and $(c)$ comes from (\ref{seteq}). This completes the proof where $e_0(v)$ is of a fixed type $P$. Now, we will show proof of an arbitrary code.
Let us denote by $C_{P}$ the set of codewords $e_0(v)$ of type $P$.
Then, as before, 
\begin{align}
&I_{max}^{\epsilon_n}(W_{e_0})\leq{I_{max}(W_{e_0,T^n_{(W)}})}   \nonumber  \\ 
									  &=\log{\sum_{z^n}\max_vW^n(z^n\vert{e_0(v))}\mathbbm{1}{((e_0(v),z^n)\in{T^n_{(W)}})}}        \nonumber \\
									  &\leq{\log{\sum_{P}\sum_{z^n}\max_{x^n\in{C_{P}}}W^n(z^n\vert{x^n)}\mathbbm{1}{((x^n,z^n)\in{T^n_{(W)}})}}} \nonumber \\
									  &\stackrel{(d)}{\leq{}}\log{\max_{P}\sum_{z^n}\max_{x^n\in{C_{P}}}W^n(z^n\vert{x^n)}\mathbbm{1}{((x^n,z^n)\in{T^n_{(W)}})}} \nonumber \\
           &+\vert{}\textbf{X}\vert{}\log{(n+1)} \nonumber \\
	&\leq{n\max_{P_X}I(P_X;W)+2n\delta{}W_c\vert{}\textbf{X}\vert{}\vert{}\textbf{Z}\vert{}+\vert{}\textbf{X}\vert{}\log{(n+1)}} \nonumber,
\end{align}  
where $(d)$ comes from Lemma \ref{countinglemma}. Notice that Theorem \ref{maxinfotheorem} is a generalization of \cite[Lemma 5]{tyagi}.
\end{proof}
\begin{rem}  \label{csiszarbound}
The typicality bounds in~\cite{info} can also be used in the proof of Theorem~\ref{maxinfotheorem}. Instead of Lemma~\ref{nsizeofsetlemma},~\cite[Lemma 2.13]{info} can be used for an upper bound on the size of a typical set. Another upper bound on the channel probability can be derived by using \cite[Lemma 2.6]{info} 
and the arguments in the proof of \cite[Lemma 2.13]{info} where the second term in the upper bound then becomes $-n\vert{}\textbf{X}\vert{}\vert{}\textbf{Y}\vert{}\delta\log{\delta}$. Then 
\begin{align*}
&I_{max}^{\epsilon_n}(W_{e_0})\leq{n\max_{P_X}I(X\wedge{}Z)-n\vert{}\textbf{X}\vert{}\vert{}\textbf{Z}\vert{}\delta\log{}(\delta^2\vert{}\textbf{X}\vert{})} \\
&+(\vert{}\textbf{X}\vert{}+\vert{}\textbf{Z}\vert{})\log{(n+1)}
\end{align*}
for every $0<\delta\leq{\frac{1}{2\vert{}\textbf{X}\vert{}\vert{}\textbf{Z}\vert{}}}$. We won't use this upper bound in our achievability analysis, which will be clear in Section~\ref{sec:Achievability_bounds}.
\end{rem}

%% file: renyi_entropy.tex
\begin{lem}  \label{maxlemma}
Let $W : \textbf{A}\rightarrow{\textbf{E}}$ be a subnormalized channel, where $\textbf{A}$ and $\textbf{E}$ are finite sets. Let $A$ be a random variable uniformly distributed on $\textbf{A}$ and $P_{E\vert{}A}=W$. Then
$$D_{2}(P_{AE}\vert{}\vert{}P_A\times{}P_E)\leq{}I_{max}(W).$$ 
\end{lem}
\begin{proof}
\begin{align}
D_{2}(P_{AE}\vert{}\vert{}P_A\times{}P_E)&=\log\sum_{a,e}\frac{(P_{AE})^2}{P_AP_E} \nonumber \\
&=\log\sum_{a,e}\frac{P_A^2(a)W^2(e\vert{}a)}{P_A(a)\sum_{a'}P_A(a')W(e\vert{}a')} \nonumber \\
						 &=\log\sum_e\frac{\sum_aW^2(e\vert{}a)}{\sum_{a'}W(e\vert{}a')} \nonumber \\
						 &{\leq{}}\log\sum_e\max_aW(e\vert{}a)=I_{max}(W)  \nonumber 
\end{align}
\end{proof}
\begin{lem} \label{D1lemma}
Let $W : \textbf{A}\rightarrow{\textbf{E}}$ be a channel, where $\textbf{A}$ and $\textbf{E}$ are finite sets. Let $A$ be a random variable uniformly distributed on $\textbf{A}$ and $P_{E\vert{}A}=W$. Then
$$D(P_{AE}\vert{}\vert{}P_A\times{}P_E)\leq{}I^\epsilon_{max}(W)-(1-\epsilon)\log(1-\epsilon)+\epsilon\log\vert{}\textbf{A}\vert{}$$
for every $\epsilon<1-e^{-1}$.
\end{lem}
\begin{proof}
Let $\textbf{T}$ be a subset of $\textbf{A}\times{}\textbf{E}$ satisfying $W(\lbrace{}e:(a,e)\in{\textbf{T}}\rbrace\vert{}a)\geq{}1-\epsilon$ for all $a\in{\textbf{A}}$. Then $D(P_{AE}\vert{}\vert{}P_A\times{}P_E)=\sum_{a,e}P_{AE}\log\frac{P_{AE}}{P_AP_E} =\sum_{(a,e)\in{\textbf{T}}}P_{AE}\log\frac{P_{AE}}{P_AP_E}+\sum_{(a,e)\in{\textbf{T}^c}}P_{AE}\log\frac{P_{AE}}{P_AP_E}.$ 
First notice that 
\begin{align*}
\log\frac{P_{AE}}{P_AP_E}&=\log\frac{P_A(a)W(e\vert{}a)}{P_A(a)\sum_{a'}P_A(a')W(e\vert{}a')} \\
						 &=\log\frac{\vert{}\textbf{A}\vert{}W(e\vert{}a)}{\sum_{a'}W(e\vert{}a')} \\
						 &\leq{}\log\frac{\vert{}\textbf{A}\vert{}W(e\vert{}a)}{W(e\vert{}a)}=\log\vert{}\textbf{A}\vert{}.
\end{align*}
Then $\sum_{(a,e)\in{\textbf{T}^c}}P_{AE}\log\frac{P_{AE}}{P_AP_E}\leq{}Pr((A,E)\in{\textbf{T}^c})\log\vert{}\textbf{A}\vert{}\leq\epsilon\log\vert{}\textbf{A}\vert{}.$ Also by \cite[Lemma 12]{moritz-boche} if $\epsilon<1-e^{-1}$, $\sum_{(a,e)\in{\textbf{T}}}P_{AE}\log\frac{P_{AE}}{P_AP_E}\leq\log{\sum_{(a,e)\in{\textbf{T}}}\frac{(P_{AE})^2}{P_AP_E}}-(1-\epsilon)\log(1-\epsilon) \leq{}I_{max}(W_{\textbf{T}})-(1-\epsilon)\log(1-\epsilon)$ where the last inequality comes from Lemma \ref{maxlemma}. So we have $D(P_{AE}\vert{}\vert{}P_A\times{}P_E)\leq{}I_{max}(W_{\textbf{T}})-(1-\epsilon)\log(1-\epsilon)+\epsilon\log\vert{}\textbf{A}\vert{}.$
Since this is true for every set $\textbf{T}$ where $W(\lbrace{}e:(a,e)\in{\textbf{T}}\rbrace\vert{}a)\geq{}1-\epsilon$ for all $a\in{\textbf{A}}$, then it holds also that $D(P_{AE}\vert{}\vert{}P_A\times{}P_E)\leq{}I^\epsilon_{max}(W)-(1-\epsilon)\log(1-\epsilon)+\epsilon\log\vert{}\textbf{A}\vert{}.$
\end{proof}

%% file: uni_hashing.tex
\begin{defn}
A random hash function $f_X$ is a stochastic map from $\textbf{A}$ to $\textbf{M}:=\lbrace1,..., M\rbrace$, where $X$ denotes a random variable describing its stochastic behavior. A random hash function $f_X$ is called an $\epsilon$-almost universal$_2$ hash function if it satisfies the following condition: For any distinct $a_1$, $a_2\in{A}$, 
$$Pr(f_X(a_1)=f_X(a_2))\leq{}\frac{\epsilon}{M}.$$
When $\epsilon=1$, we say that $f_X$ is a universal$_2$ hash function \cite{hayashi-tan}. Let \textbf{X} be the set corresponding to $X$. For every $x\in{\textbf{X}}$, $f_x^{-1} : \textbf{M}\rightarrow{\textbf{A}}$ is defined as follows: $f_x^{-1}(m)$ attains values in the preimage $f_x^{-1}\lbrace{}m\rbrace$ with uniform distribution. Notice that $f_x^{-1}$ is a right inverse, i.e., $f_x(f_x^{-1}(m))=m$. We will use universal$_2$ hash functions 
with the following extra condition as in~\cite{hayashi}: For any $x\in{\textbf{X}}$, the cardinality of $f_x^{-1}\lbrace{}m\rbrace$ does not depend on $m\in{\textbf{M}}$. Notice that the   cardinality of $f_x^{-1}\lbrace{}m\rbrace$ is $\frac{\vert{}\textbf{A}\vert{}}{\vert{}\textbf{M}\vert{}}$ for every $m\in{\textbf{M}}$. We call a universal$_2$ hash function with this extra condition a c-universal$_2$ hash function. 
\end{defn}
\begin{defn}  \label{definitionC}
\cite[Section 2-B]{hayashi-tan} Let $\textbf{A}$ and $\textbf{E}$ be two sets with finite cardinality, let $A$ be a random variable on $\textbf{A}$ and $E$ be a random variable on $\textbf{E}$. Then
\begin{align}
C_{1+s}(A\vert{}E):&=D_{1+s}(P_{AE}\vert{}\vert{}P_{mix,A}\times{}P_E)\nonumber \\
				   &=\log\vert{}\textbf{A}\vert{}-H_{1+s}(A\vert{}E) \nonumber
\end{align}
where $P_{mix,A}$ is the uniform distribution on $\textbf{A}$.
\end{defn}
\begin{lem}  \label{hashlemma}
\cite[Lemma 1]{hayashi-tan} For an $\epsilon$-almost universal$_2$ hash function $f_X : \textbf{A}\rightarrow{\textbf{M}= \lbrace1,...,M\rbrace}$, we have for $s\in[0,1]$,
$$e^{sC_{1+s}(f_X(A)\vert{}EX)}\leq{}\epsilon^s+M^se^{-sH_{1+s}(A\vert{}E)}.$$
\end{lem}

%% file: Achievability_bounds.tex
We use universal$_2$ hash functions for constructing DM-DWTC codes as in \cite{hayashi,tyagi} as follows: We have a wiretap encoder $e : \textbf{M}\rightarrow{\textbf{X}^n}$ where $e(m)=e_0(f_x^{-1}(m))$, a wiretap decoder $d :\textbf{Y}^n\rightarrow{\textbf{M}}$ where $d(y^n)=f_x(d_0(y^n))$, such that
\begin{itemize}
\item $e_0 : \textbf{A}\rightarrow{\textbf{X}^n}$ is a channel encoder,
\item $f_x^{-1} : \textbf{M}\rightarrow{\textbf{A}}$ where $f_X$ is a c-universal$_2$ hash function, 
\item $d_0 : \textbf{Y}^n\rightarrow{\textbf{A}}$ is a channel decoder,
\item $f_x : \textbf{A}\rightarrow{\textbf{M}}$ where realization $x$ is known to both receivers.
\end{itemize}
\begin{conj} \label{entropyconjecture}
Let $V : \textbf{X}\rightarrow{\textbf{Y}}$ and $W : \textbf{Y}\rightarrow{\textbf{E}}$ 
be DMC. Then, for every sequence of channel encoders $e_0[n] : \textbf{A}_n\rightarrow{\textbf{X}^n}$ which achieves the capacity of $V$ as $n\rightarrow{\infty}$, 
\begin{equation} \label{assumption}
\vert{}H(A_n\vert{}E^n)-H_{1+\frac{1}{\sqrt{n}}}(A_n\vert{}E^n)\vert{}\leq{}O(\sqrt{n}).
\end{equation}
\end{conj}
\begin{thm}  \label{achievabilitytheorem}
Let $V : \textbf{X}\rightarrow{\textbf{Y}}$ and $W : \textbf{Y}\rightarrow{\textbf{E}}$ 
be DMC. Let $e_0[n]: \textbf{A}_n\rightarrow{\textbf{X}^n}$ be a sequence of channel encoders which achieves the capacity of $V$ as $n\rightarrow{\infty}$, let $M_n$ be a random variable uniformly distributed on message set $\textbf{M}_n$ where $\vert{}\textbf{M}_n\vert{}$ divides $\vert{}\textbf{A}_n\vert{}$. Assume Conjecture \ref{entropyconjecture} holds. Then for every sequence of c-universal$_2$ hash functions $f_{X_n}$ where $f_{x_n}^{-1} : \textbf{M}_n\rightarrow{\textbf{A}_n}$ and $X_n$ is independent of $M_n$, if we let $\vert{}\textbf{M}_n\vert{}=e^{nR}$, then
\begin{equation}
\limsup_{n\rightarrow{\infty}}\frac{1}{n}C_{1}(f_{X_n}(A_n)\vert{}E^nX_n)\leq{}\vert{}R-(C_V-C_U)\vert{}^+ \nonumber
\end{equation}
where $C_V$ is the capacity of the channel $V$ and $C_U$ is the capacity of the channel $U:=W\circ{}V$. Also $C_{1}(f_{X_n}(A_n)\vert{}E^nX_n)$ goes to $0$ as $n\rightarrow{\infty}$ for $R<(C_V-C_U)$. 
Let $\vert{}\textbf{M}_n\vert{}=e^{n(C_V-C_U)+(\sqrt{n})^{1+\gamma}L}$, then
\begin{equation}
\limsup_{n\rightarrow{\infty}}\frac{1}{(\sqrt{n})^{1+\gamma}}C_1(f_{X_n}(A_n)\vert{}E^nX_n)\leq{}\vert{}L+2U_c\vert{}\textbf{X}\vert{}\vert{}\textbf{E}\vert{}\vert{}^+ \nonumber
\end{equation}
for every $0<\gamma<1$. Also $C_{1}(f_{X_n}(A_n)\vert{}E^nX_n)$ goes to $0$ as $n\rightarrow{\infty}$ for $L<-2U_c\vert{}\textbf{X}\vert{}\vert{}\textbf{E}\vert{}$.
\end{thm}
\begin{proof}
Let $f_{X_n}$ be a c-universal$_2$ hash function where $f_{x_n}^{-1} : \textbf{M}_n\rightarrow{\textbf{A}_n}$ and $X_n$ is independent of $M_n$. Notice that since the messages are uniformly distributed, $f_{x_n}^{-1} : \textbf{M}_n\rightarrow{\textbf{A}_n}$ defines a random variable $A_n$ which is uniformly distributed on $\textbf{A}_n$. This implies that $C_1(A_n\vert{}E^n)=D(P_{A_nE^n}\vert{}\vert{}P_{A_n}\times{}P_{E{^n}})$.
By Theorem \ref{maxinfotheorem}, for every $\delta>0$ it holds that
\begin{equation} \label{Imaxeq}
I^{\epsilon_n}_{max}(U_{e_0})\leq{}nC_{U}+2n\delta{}U_c\vert{}\textbf{X}\vert{}\vert{}\textbf{E}\vert{}+\vert{}\textbf{X}\vert{}\log{(n+1)}
\end{equation}
where $\epsilon_n=2\vert{}\textbf{X}\vert{}\vert{}\textbf{E}\vert{}e^{-2n\delta^2}$. Also by Lemma \ref{D1lemma}
\begin{align} \label{C1eq}
&C_1(A_n\vert{}E^n)\leq{}I^\epsilon_{max}(U_{e_0})-(1-\epsilon)\log(1-\epsilon) \nonumber \\
&+\epsilon\log\vert{}\textbf{A}_n\vert{}
\end{align}
for every $\epsilon<1-e^{-1}$. Then since $\log\vert{}\textbf{A}_n\vert{}=O(n)$, by (\ref{Imaxeq}) and (\ref{C1eq}) we have $C_1(A_n\vert{}E^n)\leq{}nC_{U}+2n\delta{}U_c\vert{}\textbf{X}\vert{}\vert{}\textbf{E}\vert{}+O(\log{}n)$ for sufficiently large $n$ and $\delta>0$ fixed. Let us analyse the expression $2n\delta{}U_c\vert{}\textbf{X}\vert{}\vert{}\textbf{E}\vert{}$. Since we don't want this expression to affect the first order coding rate, we let $\delta(n)\rightarrow{0}$ as $n\rightarrow{\infty}$ under the constraint $ne^{-2n\delta^2(n)}\leq{}O(\log{}n)$. Then we can choose $\delta(n)$ such that $\delta(n)=(\sqrt{n})^{-1+\gamma}$ where $0<\gamma<1$ and $\gamma$ is fixed. 
We also have from \cite{polyanskiy,tan} that $\log\vert{}\textbf{A}_n\vert{}\geq{}nC_{V}+O(\sqrt{n}).$
So 
\begin{align} \label{conditionalentropyeq}
&H(A_n\vert{}E^n)\geq{}n(C_V-C_{U})-2U_c\vert{}\textbf{X}\vert{}\vert{}\textbf{E}\vert{}(\sqrt{n})^{1+\gamma} \nonumber \\
&+O(\sqrt{n}). 
\end{align}
Let $\vert{}\textbf{M}_n\vert{}=e^{nR}$. By Lemma \ref{hashlemma} we have for $\epsilon=1$ and $s>0$
\begin{align} \label{firstordereq}
&\limsup_{n\rightarrow{\infty}}\frac{1}{n}C_{1+s}(f_{X_n}(A_n)\vert{}E^nX_n) \nonumber \\
&\leq{}\limsup_{n\rightarrow{\infty}}\frac{1}{n}(s^{-1}\log{}(1+M_n^se^{-sH_{1+s}(A_n\vert{}E^n)})) \nonumber \\
					                                           &=\vert{}\limsup_{n\rightarrow{\infty}}\frac{1}{n}(nR-H_{1+s}(A_n\vert{}E^n))\vert{}^+. 
\end{align}
Let $s=\frac{1}{\sqrt{n}}$. Then by \eqref{assumption}, (\ref{conditionalentropyeq}), (\ref{firstordereq}),  and since $D_{1+s}(P\Vert{}\vert{}Q)$ is monotonically increasing in $s\in[0,1]$ \cite{hayashi-tan}, 
\begin{equation} \label{nfirstordereq}
\limsup_{n\rightarrow{\infty}}\frac{1}{n}C_{1}(f_{X_n}(A_n)\vert{}E^nX_n)\leq{}\vert{}R-(C_V-C_U)\vert{}^+.
\end{equation}
Without normalizing by $n$ we also have for $s=\frac{1}{\sqrt{n}}$
\begin{align} \label{strongfirstordereq}
&C_{1+\frac{1}{\sqrt{n}}}(f_{X_n}(A_n)\vert{}E^nX_n) \nonumber \\ 
&\leq{}\sqrt{n}\log{}(1+e^{\frac{1}{\sqrt{n}}(nR-H_{1+\frac{1}{\sqrt{n}}}(A_n\vert{}E^n))}).
\end{align}
Then by (\ref{assumption}), (\ref{conditionalentropyeq}), (\ref{strongfirstordereq}), , and since $D_{1+s}(P\Vert{}\vert{}Q)$ is monotonically increasing in $s\in[0,1]$ , for sufficiently large $n$
\begin{align} \label{nstrongfirstordereq}
&C_{1}(f_{X_n}(A_n)\vert{}E^nX_n) \nonumber \\
&\leq{}\sqrt{n}\log{}(1+e^{\sqrt{n}(R-(C_V-C_U))+O((\sqrt{n})^\gamma)}).
\end{align}
This implies that $C_{1}(f_{X_n}(A_n)\vert{}E^nX_n)$ goes to $0$ as $n\rightarrow{\infty}$ for $R<(C_V-C_U)$. Let $\vert{}\textbf{M}_n\vert{}=e^{n(C_V-C_U)+(\sqrt{n})^{1+\gamma}L}$. Then similarly as in (\ref{conditionalentropyeq})--(\ref{nfirstordereq}),
\begin{align*} \label{secondordereq}
\limsup_{n\rightarrow{\infty}}\frac{1}{(\sqrt{n})^{1+\gamma}}C_1(f_{X_n}(A_n)\vert{}E^nX_n)\leq{}\vert{}L+2U_c\vert{}\textbf{X}\vert{}\vert{}\textbf{E}\vert{}\vert{}^+.
\end{align*}
Also $C_{1}(f_{X_n}(A_n)\vert{}E^nX_n)$ goes to $0$ as $n\rightarrow{\infty}$ for $L<-2U_c\vert{}\textbf{X}\vert{}\vert{}\textbf{E}\vert{}$ similarly as in (\ref{strongfirstordereq})--(\ref{nstrongfirstordereq}).
\end{proof}
\begin{rem}  \label{secondorderremark}
In the proof of Theorem~\ref{achievabilitytheorem}, we used $ne^{-2n\delta^2(n)}\leq{}O(\log{}n)$ as a constraint for comparison with the case where $\delta>0$ is fixed. However, there are two terms which depend on $\delta(n)$, namely  $n\delta(n)$ and $ne^{-2n\delta^2(n)}$. Depending on the choice of $\delta(n)$, the order of the bigger term will be the second-order, and it is easy to see that the second-order will be strictly larger than $\sqrt{n}$ independent of how $\delta(n)$ is chosen. 
Notice that we also do not use the upper bound in Remark \ref{csiszarbound} in Theorem \ref{achievabilitytheorem} because it would imply a bigger second-order. 
There is also no mention of the average error probability  $\varepsilon$ 
since $\varepsilon$ just affects the $\sqrt{n}$ order \cite{polyanskiy,tan}. 
Similarly, we could have used maximal error probability, which wouldn't have affected our results. 
\end{rem}
\begin{rem} 
Notice that the maximum first-order achievable rate in Theorem~\ref{achievabilitytheorem} is $(C_V-C_U)$, which is the weakly symmetric degraded wiretap capacity \cite[Proposition 3.2]{bloch-barros}. Furthermore, Theorem~\ref{achievabilitytheorem} can be generalized from a degraded wiretap channel to a wiretap channel where the eavesdropper's channel is less capable than the legitimate receiver, i.e., $I(P_X;V)\geq{}I(P_X;W)$ for every $P_X$ where $V : \textbf{X}\rightarrow{\textbf{Y}}$ and $W : \textbf{X}\rightarrow{\textbf{E}}$ be DMC \cite[Definition 3.12]{bloch-barros}. First, as in the proof of \cite[Theorem 4.2]{tan}, choose a type $P$ which is closest in variational distance to the capacity-achieving input distribution $P_X$, where $C=I(P_X;V)-I(P_X;W)$ \cite[Corollary 3.5]{bloch-barros}. Then use in the generalized theorem the achievability bound for constant composition codes from \cite[Theorem 4.2]{tan} together with (\ref{fixedtypeeq}) which is for fixed type similarly as in Theorem \ref{achievabilitytheorem}. Notice that the first-order rate $(I(P;V)-I(P;W))$ approaches the capacity as $n\rightarrow{\infty}$, whereas the second-order rate doesn't change.
\end{rem}